\documentclass{CSML}
\pdfoutput=1

\usepackage{lastpage}
\newcommand{\dOi}{14(2:6)2018}
\lmcsheading{}{1--\pageref{LastPage}}{}{}%
{Jul.~06,~2016}{Apr.~27, 2018}{}

\usepackage{hyperref}
\hypersetup{hidelinks}
\usepackage{amssymb}
\usepackage{graphicx,epsfig,epstopdf}

\theoremstyle{plain}\newtheorem{remark}[thm]{Remark}
\theoremstyle{plain}\newtheorem{example}[thm]{Example}
\theoremstyle{plain}\newtheorem{question}[thm]{Question}

\newcommand{\da}{\mathord{\downarrow}}
\newcommand{\ua}{\mathord{\uparrow}}
\newcommand{\bigsup}{\bigvee}
\newcommand{\biginf}{\bigwedge}
\newcommand{\aux}{\operatorname{Aux}}
\newcommand{\app}{\operatorname{App}}
\newcommand{\papp}{\operatorname{PApp}}

\newcommand{\intr}{\operatorname{int}}
\newcommand{\cl}{\operatorname{cl}}
\newcommand\twoheaduparrow{\mathord{\rotatebox[origin=c]{90}{$\twoheadrightarrow$}}}
\newcommand\twoheaddownarrow{\mathord{\rotatebox[origin=c]{90}{$\twoheadleftarrow$}}}
\newcommand{\dda}{\twoheaddownarrow}
\newcommand{\dua}{\twoheaduparrow}
\newcommand{\lra}{\longrightarrow}
\newcommand{\lap}[1]{{#1}^{\da \prec}}
\newcommand{\uap}[1]{{#1}^{\ua \prec}}

\newcommand{\lapa}[2]{#1^{\da \prec_{#2}}}
\newcommand{\uapa}[2]{#1^{\ua \prec_{#2}}}

\begin{document}

\title[Domains via approximation operators]{Domains via approximation operators}

\author[Z. Zou]{Zhiwei Zou{\rsuper{a}}}	
\author[Q. Li]{*Qingguo Li{\rsuper{b}}}	
\author[W. K. Ho]{Weng Kin Ho{\rsuper{c}}}	

\address{{\lsuper{a,b}}College of Mathematics and Econometrics,
         Hunan University, Changsha, Hunan, China 410082}	
\email{zouzhiwei1983@163.com}  
\email{liqingguoli@aliyun.com}  
\thanks{This work was supported by the National Science Foundation of China (No.11371130, 11101135) and the Higher School Doctoral Subject Foundation of Ministry of Education of China (No. 20120161110017). \\
*Corresponding author}	

\address{{\lsuper{c}}National Institute of Education, Nanyang Technological University, Singapore 637616}	
\email{wengkin.ho@nie.edu.sg}  





\begin{abstract}
  \noindent In this paper, we tailor-make new approximation operators inspired by rough set theory and specially suited for domain theory. Our approximation operators offer a fresh perspective to existing concepts and results in domain theory, but also reveal ways to establishing novel domain-theoretic results.  For instance, (1) the well-known interpolation property of the way-below relation on a continuous poset is equivalent to the idempotence of a certain set-operator; (2) the continuity of a poset can be characterized by the coincidence of the Scott closure operator and the upper approximation operator induced by the way below relation; (3) meet-continuity can be established from a certain property of the topological closure operator.  Additionally, we show how, to each approximating relation, an associated order-compatible topology can be defined in such a way that for the case of a continuous poset the topology associated to the way-below relation is exactly the Scott topology.  A preliminary investigation is carried out on this new topology.
\end{abstract}

\maketitle

\section{Introduction} \label{sec: intro}

The invention of \emph{domain theory} by Dana Scott in around the late 1960s was originally intended to give a denotational semantics for the $\lambda$-calculus~\cite{scott72a,scott72b,scott75a}.  Domain theory is a theory of approximation as it formalizes the notion of partial or incomplete information to represent computations that have not yet returned a result, and thus provides the mathematical foundation to explicate the phenomenon of approximation in computation.  Another important but entirely independent theory of approximation that has far-reaching impact, particularly, in knowledge representation and reasoning as well as in data analysis is that of \emph{rough set theory}.  Introduced by Zdis{\l}aw Pawlak in~\cite{pawlak91} as a mathematical tool for imperfect data analysis, rough set theory saw many applications in decision support occurring in various disciplines such as computer science, engineering, medicine, education and others.

Earlier works due to Guo-Qiang Zhang (\cite{zhang91,zhang06a,zhang06b}) already identified that domain theory, in the form of Scott information systems and formal concept analysis, has something to do with rough set theory because of its applications in knowledge representation and reasoning.  But it was Yi Yu Yao (\cite{yao04,yao06}) who first explicitly pointed out the role that domain theory plays in rough set theory relying on the notion of object-oriented concept. Following up on Yao's work, Lei and Luo gave a rough concept representation of complete lattices and showed that rough approximable concepts represent algebraic lattices via two key representation theorems (Theorems 3.2 and 4.6,\cite{leiluo09}).

While the aforementioned works demonstrate \emph{in one direction} that domain theory manifests in rough set theory via formal concept analysis, it has never been shown \emph{in the other direction} what r\^{o}le rough set theory plays in domain theory.  It is for this very reason that we write this paper -- we demonstrate how concepts and well-known results in domain theory can formulated and obtained via our rough set theoretic approach.  Readers who are well-acquainted with rough set theory know very well that in Pawlak's original formulation, the lower and upper approximation operators induced by an equivalence relation $R$ on a universe $\mathbb{U}$ of finite objects play a fundamental r\^{o}le. However, in domain theory, the way-below relation $\ll$ on a poset (which is of central importance) is far from being reflexive, let alone an equivalence relation.  For us, the most critical step in the adopting a rough set theoretic approach to domain theory is to get around this problem by defining a new set of approximation operators that is best suited for irreflexive binary relations on posets.

We first gather all the domain-theoretic preliminaries needed for our present development in Section~\ref{sec: prelim}, and then study how the core concepts in domain theory can be formulated using ideas from rough set theory in Section~\ref{sec: approximation operators}.  This study is carried out over two subsections: (1) We begin by defining the new approximation operators induced by an auxiliary relation on a poset in Subsection~\ref{subsec: approx operators on aux rel}.  Amongst the several order-theoretic properties that we prove concerning these operators, the highlight is a certain characterization of the interpolation property enjoyed by a given auxiliary relation in terms of the lower (respectively, upper) approximation operators.  (2) We proceed to establish more domain-theoretic results in Subsection~\ref{subsec: approx operators on approx rel} by requiring further that the auxiliary relation is (pre-)approximating.  To this achieve this, we introduce a new order-compatible topology induced by an (pre-)approximating relation.  In the case of a continuous poset, we established that the topology induced by the way-below relation is \emph{exactly} the Scott topology.  This newly introduced topology is studied in greater detail in Section~\ref{sec: further properties}.

In this paper, we characterize continuous posets to be precisely those for which the upper approximation operator induced by the way-below relation coincides with the Scott closure operator. Based on this new characterization of continuity, we obtain a novel characterization of Scott closed sets of a continuous poset--a matter which we continue to pursue further later in Section~\ref{sec: scott closure}.  This further investigation reveals some link with an earlier result (unpublished) by Dongsheng Zhao~\cite{zhao92}.  We conclude this paper in Section~\ref{sec: conclusion} by pointing the reader to possible future directions and posing some open questions.

For references, we point the reader to \cite{abramskyjung94,gierzetal03} for domain theory, and \cite{pawlak91} for rough set theory.

\section{Preliminaries} \label{sec: prelim}
A \emph{partially ordered set} (\emph{poset}, for short) is a non-empty set $P$ equipped with a reflexive, transitive and antisymmetric relation $\leq$.  If $P$ is a poset and $X \subseteq P$, we use the symbol $\bigvee X$ (respectively, $\bigwedge X$) to denote the \emph{least upper bound} or \emph{supremum} (respectively, the \emph{greatest lower bound} or \emph{infimum}) of $X$, if it exists.  We define the \emph{upper closure} of $X$ by
\[
\ua X := \{x \in P \mid \exists y \in X.~y \leq x\},
\]
and dually, $\da X$, the \emph{lower closure} of $X$.
If $X = \ua X$ (respectively, $X = \da X$), then $X$ is called an \emph{upper} (respectively, \emph{lower}) subset of $P$. It is well known that the collection of all upper (respectively, lower) subsets of a poset $P$ ordered by inclusion is a complete lattice, and we denote it by $\mathcal{U}(P)$
(respectively, $\mathcal{L}(P)$).

A subset $D$ of $P$ is \emph{directed} (dually, \emph{filtered}) provided it is non-empty and every finite subset of $D$ has an upper (dually, lower) bound in $D$.  A poset of which every directed subset has a supremum is called a \emph{directed-complete partial order} (\emph{dcpo}, for short).

In domain theory, the most important binary relation on a poset is the way-below relation, which we define as follows.  Let $P$ be a poset.  We say that $x$ is \emph{way below} $y$, denoted by $x \ll y$ if for all directed subsets $D \subseteq P$ for which $\bigsup D$ exists, the relation $\bigsup D \geq y$ always implies $x \in \da D$.  An element $x$ satisfying the condition $x \ll x$ is said to be \emph{compact}. For each $x \in P$, we use the notations
 \[
 \dda x := \{y \in P \mid y \ll x\}
 \text{ and }
 \dua x := \{y \in P \mid x \ll y\}.
 \]
It follows immediately from the definition of $\ll$ that
$\dda x$ (respectively, $\dua x$) is a lower (respectively, upper) subset of $P$.   A poset $P$ is said to be \emph{continuous} if for any $x \in P$, the set $\dda x$ is directed and
$x = \bigsup \dda x$.

The following proposition often comes in handy to prove that a poset is continuous.
\begin{prop} \label{prop: basis lemma}\cite{gierzetal03}
Let $P$ be a poset.  If for each $x \in P$ there exists a directed set $D_x \subseteq \dda x$ with $\bigsup D_x = x$, then \begin{enumerate}
\item $\dda x$ is directed, and
\item $\bigsup \dda x = x$.
\end{enumerate}
In particular, $P$ is continuous.
\end{prop}

A binary relation $\prec $ on a poset $P$ is called an \emph{auxiliary relation} if it satisfies the following conditions for any $u,~x,~y$ and $z \in P$:
\begin{enumerate}
\item $x \prec y$ implies $x \leq y$;
\item $u \leq x \prec y \leq z$ implies $u \prec z$;
\item if the smallest element $\bot$ of $P$ exists, then $\bot \prec x$.
\end{enumerate}
The set of all auxiliary relations on $P$ is denoted by $\aux(P)$.
Clearly, both $\leq$ and $\ll$ are auxiliary relations on $P$.

Let ${\prec}\in \aux(P)$ and $x \in P$.
We denote the set $\{y \in P \mid y \prec x\}$ by $s_\prec (x)$.
From the definition of auxiliary relation, we have that $s_{\prec}(x) \in \mathcal{L}(P)$.

The set $\aux(P)$ is a complete lattice with respect to the containment of graphs as subsets of $P \times P$ since it is closed under arbitrary intersections and unions in $P \times P$.

An auxiliary relation $\prec$ on a poset $P$ is called \emph{pre-approximating} if for any $x \in P$, $s_{\prec}(x)$ is directed. If in addition $x = \bigsup s_\prec (x)$ holds for any $x \in P$, then we say that $\prec$ is \emph{approximating}. The set of all pre-approximating (respectively, approximating) auxiliary relations is denoted by $\papp(P)$ (respectively, $\app(P)$).  For a poset $(P,\leq)$, the order relation $\leq$ is trivially approximating, and the way below relation $\ll$ is approximating if and only if $P$ is continuous.

\begin{prop} \label{prop: cont implies way-below is smallest approx}\cite{gierzetal03}
In a poset $P$, the way below relation $\ll$ is contained in all approximating auxiliary relations, i.e.,
\[
\ll ~\subseteq ~\bigcap \{\prec \mid {\prec}\in \app(P)\}.
\]
If $P$ is a continuous poset, then
\[
\ll ~=~ \bigcap \{\prec \mid {\prec}\in \app(P)\}.
\]
Thus, the following are equivalent for a poset $P$:
\begin{enumerate}
\item $P$ is continuous.
\item The way-below relation $\ll$ is the smallest approximating relation on $P$.
\end{enumerate}
\end{prop}

An auxiliary relation $\prec$ on a poset $P$ is said to satisfy the \emph{Interpolation Property} (INT), provided that for all $x$ and $z \in P$ we have
\begin{equation} \tag{INT}
(x \prec z) \implies (\exists y \in P.~ x \prec y \prec z).
\end{equation}

In domain theory, much depends on the fact that $\ll$ satisfies the interpolation property.
\begin{prop} \label{prop: cont implies way-below is INT}\cite{gierzetal03}
Let $P$ be a continuous poset.  Then $\ll$ satisfies (INT).
\end{prop}

The most prominently featured topology in domain theory is the \emph{Scott topology}. Let $P$ be a poset and $U \subseteq P$.  We say that $U$ is \emph{Scott open} if
\begin{enumerate}
\item $U = \ua U$;
\item $U$ is inaccessible by directed suprema, i.e., whenever a directed subset $D$ is such that $\sup D \in U$, there already exists $d \in D$ such that $d \in U$.
\end{enumerate}
A set is \emph{Scott closed} if it is the complement of a Scott open set.
It follows immediately that a subset $A$ of $P$ is Scott closed if and only if it is lower and closed under the formation of directed suprema.
The collection of all Scott open subsets of $P$ forms a topology on $P$ called the \emph{Scott topology} and is denoted by $\sigma (P)$.

As a direct consequence of Proposition~\ref{prop: cont implies way-below is INT}, we have:
\begin{prop} \label{prop: Scott topology base for cont}\cite{gierzetal03}
Let $P$ be a continuous poset and $x \in P$.
\begin{enumerate}
\item The set $\dua x$ is Scott open.
\item The collection $\mathcal{B} := \{\dua x \mid x \in P\}$ forms a base for the Scott topology on $P$.
\item The Scott interior of $\ua x$ is precisely $\dua x$, i.e., $\intr_\sigma(\ua x) = \dua x$.
\end{enumerate}
\end{prop}

\section{Approximation operators} \label{sec: approximation operators}

We can define many kinds of operators for domain theory. In this paper, we aim to get "approximation" operators. Here approximation operator means that for each subset of a poset, it is always captured between its lower and upper approximations.

\subsection{Auxiliary relation}
\label{subsec: approx operators on aux rel}
In this subsection, we define a new set of approximation operators for domain theory.
\begin{defi}[Lower and upper approximations] \label{defi: lower and upper approx in dom}
Let $P$ be a poset and ${\prec}\in \aux(P)$.  For each subset $A$ of $P$, the lower and upper approximations of $A$ (with respect to $\prec$), denoted by $\lap{A}$ and $\uap{A}$ respectively, are defined by
\[
\lap{A} := \{x \in P \mid x \in A \ \& \ s_\prec(x) \cap A \neq \emptyset \}
\]
and
\[
\uap{A} := \{x \in P \mid s_\prec(x) \subseteq \da A\}.
\]
\end{defi}

The following are immediate consequences arising from Definition~\ref{defi: lower and upper approx in dom}:
\begin{prop}
\label{prop: basic uap and lap}
Let $P$ be a poset, ${\prec}\in \aux(P)$, $A \subseteq P$ and $a \in P$.  Then the following properties hold:
\begin{enumerate}
\item $\uap{A} = \uap{(\da A)}$.
\item $A^{\ua \leq} = \da A$ and $A^{\da \leq} = A$.
\item $(\ua a)^{\da \prec} = \{x \in P \mid a \prec x\}$.
\item For any set $A \subseteq P$, $x \in \lap{A}$ if and only if there exists $y \in A$ such that $y \prec x$.
\item $\uap{A} \in \mathcal{L}(P)$.
\item If $A \in \mathcal{U}(P)$, then $\lap{A} \in \mathcal{U}(P)$.
\end{enumerate}
\end{prop}

Our choice of definition for the lower and upper approximations is justified in view of the following proposition:
\begin{prop}
\label{prop: fundamental approx condition}
Let $P$ be a poset.  For each $A \subseteq P$ and ${\prec}\in \aux(P)$, we have
\[
\lap{A} \subseteq A \subseteq \uap{A}.
\]
\end{prop}

Given a poset $P$ and ${\prec}\in \aux(P)$, we always have
\[
\lap{\emptyset} = \emptyset \text{ and } \uap{P} = P.
\]
It is very tempting to conclude that $\uap{\emptyset} = \emptyset$.  But this is not true in general.
For instance, if $P$ has no bottom element, there may still be an element $x \in P$ such that $s_\prec(x) = \emptyset$.
The following proposition tells us more about this situation.
\begin{prop}
\label{prop: lap of whole space}
Let $P$ be a poset and ${\prec}\in \aux(P)$.  Then the following are equivalent:
\begin{enumerate}
\item For any $x \in P$, the set $s_\prec(x)$ is nonempty.
\item $\uap{\emptyset} = \emptyset$.
\item $\lap{P} = P$.
\end{enumerate}
\end{prop}
\begin{proof}
(1) $\implies$ (2): It is obvious by Definition~\ref{defi: lower and upper approx in dom}.

(2) $\implies$ (3): Suppose that $\lap{P} \neq P$.  Then there exists $x \in P$ such that $x \not \in \lap{P}$.
This means that $s_\prec(x) \cap P = \emptyset$.
Thus, $s_\prec(x) = \emptyset$.  This means that $s_\prec(x) \subseteq \da \emptyset$, and hence $\lap{\emptyset} \neq \emptyset$ because $x \in \lap{\emptyset}$, and thus (2) cannot hold.

(3) $\implies$ (1): Suppose for the sake of contradiction that there is $x \in P$ such that $s_\prec(x) = \emptyset$.
Then we have $s_\prec(x) \cap P = \emptyset$.
This implies that there exists $x \in P$ for which the condition $s_\prec(x) \cap P \neq \emptyset$ fails to hold, i.e., there is $x \in P \backslash \lap{P}$, and thus (3) cannot hold.
\end{proof}

We now turn the lower and upper approximations into set-valued operators.
Let $P$ be a poset, and we denote the collection of upper (respectively, lower) subsets of $P$ by $\mathcal{U}(P)$ (respectively, $\mathcal{L}(P)$.  Recall that $(\mathcal{U}(P),\subseteq)$ and $(\mathcal{L}(P),\subseteq)$ are complete lattices.  In each case, infima are calculated as intersections and suprema as unions.

It is easy to establish the following:

\begin{prop}
\label{cor: uap preserves intersection}
Let $P$ be a poset and ${\prec}\in \aux(P)$.
Then the following hold:
\begin{enumerate}
\item For any collection $\{A_i \mid i \in I\}$ of lower subsets of $P$,
\[
\uap{\left(\bigcap_{i \in I} A_i \right)}
= \bigcap_{i \in I} \uap{A_i}.
\]

\item For any collection $\{B_i \mid i \in I\}$ of upper subsets of $P$,
\[
\lap{\left(\bigcup_{i \in I} B_i \right)}
= \bigcup_{i \in I} \lap{B_i}.
\]
\end{enumerate}
\end{prop}

Considering the mappings $\uap{(-)}: \mathcal{L}(P) \lra \mathcal{L}(P),
~ A \mapsto \uap{A}$ and $\lap{(-)}: \mathcal{U}(P) \lra \mathcal{U}(P),
~ A \mapsto \lap{A}$, they are both very special monotone mappings in that they are part of certain adjunctions. At this juncture, the reader may like to recall that a pair of monotone maps $f:S \lra T$ and $g:T \lra S$ between posets $S$ and $T$ is called an \emph{adjunction} if for each $s \in S$ and $t \in T$, $f(s) \geq t$ if and only if $s \geq g(t)$.  In an adjunction $(f,g)$, the function $f$ is called the upper adjoint and $g$ the lower adjoint.

\begin{lem}\cite{gierzetal03}
\begin{enumerate}
\item Let $f:S \lra T$ be a function between posets of which $S$ is a complete lattice. Then $f$ preserves infs iff $f$ is monotone and has a lower adjoint.
\item Let $g:T \lra S$ be a function between posets of which $T$ is a complete lattice. Then $g$ preserves sups iff $g$ is monotone and has an upper adjoint.
\end{enumerate}
\end{lem}

\begin{cor}
\label{thm: uap is an upper adjoint}\leavevmode
\begin{enumerate}
\item The monotone map
\[
\uap{(-)}: \mathcal{L}(P) \lra \mathcal{L}(P),
~ A \mapsto \uap{A}
\]
has a lower adjoint.
\item The monotone mapping
\[
\lap{(-)}: \mathcal{U}(P) \lra \mathcal{U}(P),
~ A \mapsto \lap{A}
\]
has an upper adjoint.
\end{enumerate}
\end{cor}

The upper and lower approximations are closely connected to each other.

Let $P$ be a poset and $\prec$ an auxiliary relation on $P$.  For any $A \subseteq P$, we have:
\[
\begin{split}
      & x \in \lap{A} \\
 \iff & (x \in A) \land (s_\prec(x) \cap A \neq \emptyset) \\
 \iff & (x \in A) \land \neg(s_\prec(x) \cap A = \emptyset) \\
 \iff & (x \in A) \land \neg(s_\prec(x) \subseteq P\backslash A).
\end{split}
\]
Observe that $s_\prec(x) \subseteq P \backslash A$ always implies $x \in \uap{(P \backslash A)}$ because we always have $(P \backslash A) \subseteq \da (P \backslash A)$.
Thus, contrapositivity forces that
\[
\neg(x \in \uap{(P \backslash A)}) \implies
\neg(s_\prec(x) \subseteq P \backslash A)
\]
which in turn entails that
\[
(x \in A) \land \neg(x \in \uap{(P \backslash A)})
\implies
(x \in A) \land \neg(s_\prec(x) \subseteq P \backslash A),
\]
i.e.,
\[
A \cap (P \backslash \uap{(P \backslash A)}) \subseteq \lap{A}.
\]

In the event that $A$ is an upper set, then $P \backslash A$ is lower so that $\da (P \backslash A) = P \backslash A$.
Thus, we have
\[
\lap{A} = A \cap (P \backslash \uap{(P \backslash A)})
\]
if $A$ is an upper set.

The above discussion establishes the following partition property:
\begin{thm}
\label{thm: uap and lap determines each other}
Let $P$ be a poset, ${\prec}\in \aux(P)$ and $A \subseteq P$.
Then the following hold:
\begin{enumerate}
\item $\lap{A} \cup \uap{(P \backslash A)} = P$.
\item If in addition $A$ is an upper set of $P$, then
\[
\lap{A} \cap \uap{(P \backslash A)} = \emptyset.
\]
That is, $\lap{A}$ and $\uap{(P \backslash A)}$ define a partition on $P$.
\end{enumerate}
\end{thm}

One advantage of our operators is that it allows us to characterize the interpolating property of an auxiliary relation $\prec$.  More precisely, an auxiliary relation is interpolating if and only if the upper (respectively, lower) approximation operator is a closure (respectively, kernel) operator on the lattice of lower (respectively, upper) sets. Here, a \emph{closure} (respectively, \emph{kernel}) \emph{operator} refers to an idempotent, monotone self map $f$ on a poset $L$ with $1_L \leq f$ (respectively, $f \leq 1_L$).

\begin{thm}
\label{thm: char of INT}
Let P be a poset and ${\prec}\in \aux(P)$.
Then the following are equivalent:
\begin{enumerate}
\item $\prec$ satisfies (INT).

\item $\lap{(\lap{A})} = \lap{A}$ for every upper set $A$ of $P$.
\item The lower approximation operator is a kernel operator on the lattice, ${(\mathcal{U}(P),\subseteq)}$, of upper sets of $P$.
\item $\uap{(\uap{B})} = \uap{B}$ for every lower set $B$ of $P$.
\item The upper approximation operator is a closure operator on the lattice, $(\mathcal{L}(P),\subseteq)$, of lower sets of $P$.
\end{enumerate}
\end{thm}
\begin{proof}
We only need to prove the equivalence of (1), (2) and (4).

(1) $\implies$ (2): Let $A \in \mathcal{U}(P)$.  By virtue of Proposition~\ref{prop: fundamental approx condition}, it suffices to show that $\lap{A} \subseteq \lap{(\lap{A})}$.  To this end, let $x \in \lap{A}$. Then there exists $y \in A$ such that $y \prec x$.
By the (INT) property, there exists $z \in P$ such that $y \prec z \prec x$.  Since $A$ is upper and $y \in A$, we have $z \in A$.
Hence $y \in \lap{A}$, and this implies $x \in \lap{(\lap{A})}$.

(2) $\implies$ (1): Let $x$ and $y$ be elements of $P$ such that $x \prec y$.  By Proposition~\ref{prop: basic uap and lap}(3), it follows that $y \in \lap{(\ua x)}$.  Note that $\ua x \in \mathcal{U}(P)$ so that by (2), we have that $\lap{(\ua x)} = \lap{(\lap{(\ua x)})}$.  So, $y \in \lap{(\lap{(\ua x)})}$.  Thus, $s_\prec(y) \cap \lap{(\ua x)} \neq \emptyset$, and so there exists $z \prec y$ with $z \in \lap{(\ua x)}$.  Then there is $w \prec z$ such that $w \in \ua x$.  Consequently, we have $x \leq w \prec z \prec y$ which implies that $x \prec z \prec y$ for some $z \in P$.  Thus, $\prec$ satisfies the (INT) property.

(2) $\iff$ (4): Immediate result of two applications of Theorem~\ref{thm: uap and lap determines each other}.
\end{proof}


On a poset $P$, the set $\aux(P)$ of auxiliary relations on $P$ is a complete lattice relative to the containment of graphs as subsets of $P \times P$.  We now develop certain basic properties of approximation operators arising from different auxiliary relations on $P$.

\begin{prop}
\label{prop: order in aux affect lap and uap}
Let $P$ be a poset and $\prec_1 \subseteq \prec_2$ in $\aux(P)$.
Then the following containments hold for any subset $A$ of $P$:
\[
\lapa{A}{1} \subseteq \lapa{A}{2} \text{ and }
\uapa{A}{1} \supseteq \uapa{A}{2}.
\]
\end{prop}

\begin{remark}\leavevmode
\begin{enumerate}
\item In view of Proposition~\ref{prop: fundamental approx condition} that
\[
\lap{A} \subseteq A \subseteq \uap{A},
\]
as $\prec$ increases in $\aux(P)$, the upper and lower approximations become closer to $A$, i.e., they give a better estimation of $A$.
In the extreme case when $\leq$, the largest element of $\aux(P)$ is taken, we have $A^{\da \leq} = A$ and $A^{\ua \leq} = \da A$.

\item Let $A$ be a fixed subset of $P$.  Then we can talk about the aforementioned effect of auxiliary relations in $\aux(P)$ on the corresponding approximation operators by viewing this cause-and-effect phenomenon as monotone mappings.  More precisely, the mappings defined below
\[
u: (\aux(P),\subseteq) \lra (\mathcal{P}(P),\supseteq),~
\prec \mapsto \uap{A}
\]
and
\[
l: (\aux(P),\subseteq) \lra (\mathcal{P}(P),\subseteq),~
\prec \mapsto \lap{A}.
\]
are monotone mappings between the complete lattices $(\aux(P),\subseteq)$ and $(\mathcal{P}(P),\supseteq)$ (respectively, $(\mathcal{P}(P),\subseteq)$).
\end{enumerate}
\end{remark}

Indeed, the monotone mappings $u$ and $l$, defined above in the remark, are lattice homomorphisms in the following sense:
\begin{prop}
Let $P$ be a poset, $A \subseteq P$ and $\prec_1, \prec_2 \in \aux(P)$.  Then the following hold:
\begin{enumerate}
\item $\uapa{A}{1} \cap \uapa{A}{2} = A^{\ua (\prec_1 \cup \prec_2)}$.

\item $\lapa{A}{1} \cup \lapa{A}{2} = A^{\da (\prec_1 \cup \prec_2)}$.

\item If $A$ is filtered, then $\lapa{A}{1} \cap \lapa{A}{2} = A^{\da (\prec_1 \cap \prec_2)}$.
\end{enumerate}
\end{prop}
\begin{proof}
Notice that for each $x \in P$, we always have:
\[
s_{\prec_1 \cup \prec_2}(x) = s_{\prec_1}(x) \cup s_{\prec_2}(x)
\]
and
\[
s_{\prec_1 \cap \prec_2}(x) = s_{\prec_1}(x) \cap s_{\prec_2}(x).
\]
Thus, (1) and (2) follow immediately.

For (3), if $x \in \lapa{A}{1} \cap \lapa{A}{2}$, then there are elements $y_1,~y_2 \in A$ such that $y_1 \prec_1 x$ and $y_2 \prec_2 x$.  Since $A$ is filtered, there exists $z \in A$ such that $z \leq y_1$ and $z \leq y_2$.  Thus, $z \prec_1 x$ and $z \prec_2 x$, and thus, $x \in A^{\da (\prec_1 \cap \prec_2)}$.
\end{proof}

\subsection{Approximating relations}
\label{subsec: approx operators on approx rel}
In this subsection, we make use of the operators introduced in Definition~\ref{defi: lower and upper approx in dom} in the situation when the auxiliary relation is an approximating one.  In particular, we study the relationships between the approximation operators derived from an approximating auxiliary relation and the Scott closure and interior operators.  Crucially, we rely on a new topology defined from the given approximating auxiliary relation.  Consequently, we obtain a novel characterization of the continuity of posets.

Let us begin with an interesting observation.
\begin{prop} \label{prop: ll open sets}
Let $P$ be a poset and $U \subseteq P$.
If $U$ is upper and $U = U^{\da \ll}$, then $U$ is Scott open.
\end{prop}
\begin{proof}
Since $U$ is upper, it suffices to show that $U$ is inaccessible by existing directed suprema. To this end, let $D$ be a directed subset whose supremum, $\bigsup D$, exists and belongs to $U$.  Because $U = U^{\da \ll}$, there exists $x \in U$ such that already $x \ll \bigsup D$.
This implies that there is $d \in D$ such that $d \geq x$.  Since $U$ is upper, $d \in U$.
\end{proof}

This observation leads us to consider, for the case of an arbitrary ${\prec}\in \aux(P)$, those upper subsets $U$ for which $U = \lap{U}$.  In general, these subsets do not form a topology on $P$ but they do when $\prec$ is pre-approximating.

\begin{defi}
Let $P$ be a poset and ${\prec}\in \papp(P)$. A subset $U$ of $P$ is said to be \emph{$\prec$-open} if it satisfies the following two conditions:
\begin{enumerate}
\item $U = \ua U$;
\item $U = \lap{U}$.
\end{enumerate}
\end{defi}

\begin{thm} \label{thm: new topology generated by lap}
Let $P$ be a poset and ${\prec}\in \papp(P)$.
Then the collection of $\prec$-open subsets defines a topology on $P$.
\end{thm}
\begin{proof}
Recall that we have earlier remarked that $\lap{\emptyset} = \emptyset$.  Since for each $x \in P$, the set $\{y \in P \mid y \prec x\}$ is directed and hence non-empty, by Proposition~\ref{prop: lap of whole space}(3), it follows that $\lap{P} = P$.  Thus, $\emptyset$ and $P$ are $\prec$-open sets.

Let $\{U_i \mid i \in I\}$ be any collection of $\prec$-opens.  We want to show that
$\bigcup_{i \in I} U_i$ is $\prec$-open.  For each $i \in I$, $U_i = \lap{U_i}$.  So we aim to show that $\bigcup_{i \in I} \lap{U_i}$ is $\prec$-open.  But by Theorem~\ref{cor: uap preserves intersection}(2), $\bigcup_{i \in I} \lap{U_i} = \lap{\left(\bigcup_{i \in I} U_i \right)}$.  Thus, $\bigcup_{i \in I} \lap{U_i} = \lap{\left(\bigcup_{i \in I} \lap{U_i}\right)}$ which completes the argument that $\bigcup_{i \in I} U_i$ is $\prec$-open.

Let $U_1$ and $U_2$ be $\prec$-opens.  We now prove that $U_1 \cap U_2$ is again $\prec$-open. It suffices to show that $\lap{(U_1 \cap U_2)} = U_1 \cap U_2$ since we have $U_1 = \lap{U_1}$ and $U_2 = \lap{U_2}$.  To achieve this, we rely on Proposition~\ref{prop: basic uap and lap}(4).  Let $x \in U_1 \cap U_2$.  This implies that $x \in U_1 = \lap{U_1}$ and $x \in U_2 = \lap{U_2}$.  Thus, by Proposition~\ref{prop: basic uap and lap}(4), there exist $y_1 \in U_1$ and $y_2 \in U_2$ such that $y_1 \prec x$ and $y_2 \prec x$.
Thus, $y_1$ and $y_2 \in s_\prec(x)$.  Since $s_\prec(x)$ is directed by assumption, there exists $y \in s_\prec(x)$ such that $y_1,~y_2 \leq y$.
But $U_1$ and $U_2$ are upper, and thus, $y \in U_1 \cap U_2$ with $y \prec x$.  By Proposition~\ref{prop: basic uap and lap}(4), it follows that $\lap{(U_1 \cap U_2)} = U_1 \cap U_2$.
\end{proof}

\begin{defi}[Topology generated by lower approximation] \label{defi: prec-topology}
Let $P$ be a poset and ${\prec}\in \papp(P)$.
We denote the collection of $\prec$-opens of $P$ by $\mu^{\prec} (P)$.  We call $\mu^\prec(P)$ the \emph{topology generated by lower approximation} or the \emph{$\mu^\prec$-topology} and the corresponding topological space, $P_{\mu^\prec} := (P,\mu^\prec(P))$, the \emph{lower approximation space}\footnote{The same nomenclature is used, with a different meaning from ours, by some authors in the literature on rough set theory, e.g., see~\cite{malikmordeson02}.}.  Whenever, we mention the $\prec$-opens of $P$, we implicitly require that $\prec$ is a pre-approximating auxiliary relation on $P$.
\end{defi}

Proposition~\ref{prop: ll open sets} already hints at a close relationship between the $\mu^\prec$-topology and the Scott topology.  In our ensuing development, we carry out our investigation in this direction.

Indeed, a $\prec$-open set bears some semblance to a Scott open set in the following sense:
\begin{prop}
\label{prop: inaccessible by joins of sprec(x)}
Let $P$ be a poset, $U \subseteq P$ and ${\prec}\in \papp(P)$.
Consider the following statements:
\begin{enumerate}
\item $U \in \mu^\prec(P)$.
\item $U$ is upper and inaccessible by existing suprema of directed sets of the form $s_\prec(x)$.
\end{enumerate}
Then (1) $\implies$ (2). \\
If, in addition ${\prec}\in \app(P)$, then (2) $\implies$ (1).
\end{prop}
\begin{proof}
(1) $\implies$ (2): Suppose $x \in P$ is such that $\bigsup s_\prec(x)$ exists and belongs to $U$.  Since $U = \lap{U}$, there exists $y \in U$ such that $y \prec \bigsup s_\prec(x) \leq x$ and so $y \prec x$.

Assume further that ${\prec}\in \app(P)$.
Then for each $x \in U$, $s_\prec(x)$ is directed and $\bigsup s_\prec(x) = x$.  By (2), there exists $y \in s_\prec(x)$ such that $y \in U$, i.e., $x \in \lap{U}$.  So, (2)~$\implies$~(1).
\end{proof}

The above immediately yields:
\begin{cor} \label{cor: induced is finer than scott}
Let $P$ be a poset and ${\prec}\in \app(P)$.
Then the $\mu^\prec$-topology is finer than the Scott topology, i.e.,
$\sigma(P) \subseteq \mu^\prec(P)$.
\end{cor}

\begin{prop} \label{prop: order compatability of mu prec}
Let $P$ be a poset and ${\prec}\in \app(P)$.
Then the $\mu^\prec$-topology is order-compatible, i.e., its specialization order coincides with the underlying order.
\end{prop}
\begin{proof}
Since the $\mu^\prec$-topology is finer than the Scott topology, its specialization order is contained in that of the Scott topology, which is the underlying order.  The reverse containment holds by virtue of the fact that $\prec$-opens are upper sets.
\end{proof}

\begin{remark}
\label{thm: scott topology coincide with induced top by ll}
The $\mu^\ll$-topology and Scott topology on a continuous poset coincide.
\end{remark}
\begin{proof}
Immediate from Proposition~\ref{prop: ll open sets} and~\ref{cor: induced is finer than scott}.
\end{proof}

\begin{lem}
\label{lem: lower inclusion}
Let $P$ be a poset, $A \subseteq P$ and ${\prec}\in \app(P)$.  Then
\[
\intr_{\mu^\prec}(A) \subseteq \lap{A}.
\]
\end{lem}
\begin{proof}
Follows from $\intr_{\mu^\prec}(A) \subseteq A$ and the monotonicity of the lower approximation operator on the power-set lattice.
\end{proof}

\begin{lem}
\label{lem: upper inclusion}
Let $P$ be a poset, $A \subseteq P$ and ${\prec}\in \app(P)$.  Then
\[
\uap{A} \subseteq \cl_{\mu^\prec}(A).
\]
\end{lem}
\begin{proof}
We first prove that for each subset $B\subseteq P$, it is $\mu^\prec$-closed iff it is lower and closed under suprema of set of the form $s_\prec(a)$, $a\in P$. Suppose $B$ is $\mu^\prec$-closed and there exists $a\in P$ such that $s_\prec(a)\subseteq B$, then $a\in B$ because $P \backslash B$ is $\mu^\prec$-open.
The contrary is obvious.

For each $x\in \uap{A}$, $s_\prec(x) \subseteq {\downarrow}A\subseteq \cl_{\mu^\prec}(A)$. Then $x\in \cl_{\mu^\prec}(A)$.
\end{proof}

The following standard topological result comes in handy soon:
\begin{lem}
\label{thm: intr and cl}
For any topological space $(X,\tau)$ and $A \subseteq X$, it holds that
\[
X \backslash \cl_\tau(A) =\intr_\tau (X \backslash A)
\]
and hence $\cl_\tau(A)$ and $\intr_\tau (X \backslash A)$ define a partition on $X$.
\end{lem}

\begin{thm} \label{thm: chain of containments}
Let $P$ be a poset, $A \subseteq P$ and ${\prec}\in \app(P)$.  Then the following chain of set inclusions hold:
\[
\intr_\sigma(A) \subseteq \intr_{\mu^\prec} (A) \subseteq \lap{A} \subseteq A \subseteq \uap{A} \subseteq \cl_{\mu^\prec}{(A)} \subseteq \cl_\sigma(A).
\]
\end{thm}
\begin{proof}
That $\intr_{\mu^\prec} (A) \subseteq \lap{A} \subseteq A \subseteq \uap{A}$
holds follows from Lemmas~\ref{lem: lower inclusion}, \ref{lem: upper inclusion} and
\ref{prop: fundamental approx condition}.  Notice that if we can establish that
$\intr_\sigma (A) \subseteq \intr_{\mu^\prec} (A)$ holds for any subset $A$ of $P$, then it follows from Lemma~\ref{thm: intr and cl} and Theorem~\ref{thm: uap and lap determines each other} that $\cl_{\mu^\prec}(A) \subseteq \cl_\sigma(A)$.  Thus, it remains to prove that $\intr_\sigma (A) \subseteq \intr_{\mu^\prec} (A)$.  But this is immediate by Lemma~\ref{cor: induced is finer than scott}.
\end{proof}

\begin{lem}
\label{lem: cont implies coincidence}
If $P$ is a continuous poset and $A \in \mathcal{U}(P)$, then
\[
\intr_\sigma(A) = A^{\da \ll}.
\]
\end{lem}
\begin{proof}
Since $P$ is a continuous poset, $\ll$ enjoys the (INT) property.  Applying Theorem~\ref{thm: char of INT} for the upper set $A$, we have $(A^{\da \ll})^{\da \ll} = A^{\da \ll}$.  Thus, $A^{\da \ll}$ is a $\ll$-open set that is contained in $A$, and so $A^{\da \ll} \subseteq \intr_{\mu^\ll}(A)$.  But Theorem~\ref{thm: chain of containments} asserts that $\intr_{\mu^\ll}(A) \subseteq A^{\da \ll}$ so that $\intr_{\mu^\ll}(A) = A^{\da \ll}$.
Finally, $P$ is continuous and so Theorem~\ref{thm: scott topology coincide with induced top by ll} asserts that $\intr_{\mu^\ll}(A) = \intr_{\sigma}(A)$.  Thus, the proof is complete.
\end{proof}

\begin{lem} \label{lem: converse lemma}
If $P$ is a poset such that every $A \in \mathcal{U}(P)$ satisfies the equation
\[
A^{\da \ll} = \intr_\sigma(A),
\]
then $P$ is continuous.
\end{lem}
\begin{proof}
Let $x \in P$ be given.  We need to show that
\begin{enumerate}
\item $\dda x$ is a directed set; and
\item $\bigsup \dda x = x$.
\end{enumerate}
Let $a,~ b \in \dda x$ be given.
Thus, $x \in \dua a$ and $x \in \dua b$ so that $x \in \dua a \cap \dua b$.  By Proposition~\ref{prop: basic uap and lap}(3), $(\ua a)^{\da \ll} = \dua a$ and
$(\ua b)^{\da \ll} = \dua b$.
By the given assumption, $\dua a = (\ua a)^{\da \ll} = \intr_\sigma(\ua a)$
and $\dua b = (\ua b)^{\da \ll} = \intr_\sigma(\ua b)$. This implies that
\[
x \in \dua a \cap \dua b = \intr_\sigma (\ua a) \cap \intr_\sigma (\ua b) = \intr_\sigma (\ua a \cap \ua b) = (\ua a \cap \ua b)^{\da \ll}.
\]
Thus there exists $c \in \ua a \cap \ua b$ such that $c \ll x$.
This proves that $\dda x$ is a directed set.

We now show that $\bigsup \dda x = x$.  Suppose for the sake of contradiction that $\bigsup \dda x \neq x$.  Because $x$ is an upper bound of $\dda x$, this would mean that there is another upper bound $y$ of $\dda x$ such that $x \not \leq y$, i.e., $x \in P \backslash \da y$.  Since $\da y$ is a Scott closed subset, the set $P \backslash \da y$ is Scott open, and hence equals to its Scott interior.
Thus, $x \in \intr_\sigma(P \backslash \da y) = (P \backslash \da y)^{\da \ll}$
by the assumption.  Therefore there is an element $w \in P \backslash \da y$ such that $w \ll x$.  This contradicts that $y$ is an upper bound of the set $\dda x$ since $w$ belongs to $\dda x$ and yet $w \not \leq y$.  Then the poset $P$ is continuous.
\end{proof}

\begin{lem} \label{lem: implication of three conditions}
If $P$ is a poset with an approximating auxiliary relation $\prec$ satisfying $\lap{A} = \intr_\sigma(A)$ for every $A \in \mathcal{U}(P)$, then the following hold:
\begin{enumerate}
\item $\mu^\prec(P) = \sigma(P)$.
\item $\prec = \ll$.
\item $P$ is continuous.
\end{enumerate}
\end{lem}
\begin{proof}\leavevmode
\begin{enumerate}
\item Note that $\lap{A} = \intr_{\sigma}(A)$ for every upper set $A$ implies that
\[
\lap{(\lap{A})} = \intr_\sigma(\intr_\sigma A) = \intr_\sigma (A) = \lap{A}.
\]
Thus, $\lap{A}$ is $\prec$-open, i.e., $\lap{A} = \intr_{\mu^\prec}(A)$.
So, for any $\prec$-open set $U$, since it is upper by definition, we have that
$U = \lap{U} = \intr_\sigma(U)$, which implies that every $\prec$-open set is Scott open.  This proves (1).

\item Note that by Proposition~\ref{prop: cont implies way-below is smallest approx}, $\ll \subseteq \prec$ for any approximating auxiliary relation $\prec$.  Let $x$ and $y \in P$ be such that $x \prec y$.  We aim to show that $x \ll y$.  Let $D$ be a directed set whose supremum exists, and $\bigsup D \geq y$.  Note that $\{w \in P \mid x \prec w\} = \lap{(\ua x)} = \intr_\sigma(\ua x)$ is a Scott open set that contains $y$.  Since Scott open sets are upper, $\bigsup D \geq y$ implies that $\bigsup D \in \{w \in P \mid x \prec w\}$.  Thus, there exists $d \in D$ such that $x \prec d$, and hence $x \leq d$.
Thus, $\prec = \ll$.

\item By the preceding part, $\prec = \ll$ which is approximating.  So, by definition, $P$ is continuous.
\end{enumerate}
\end{proof}

All in all, we obtained a new characterization of continuity of posets:
\begin{thm}[A characterization of continuous posets]
\label{thm: cha of cont posets}
The following are equivalent for a poset $P$.
\begin{enumerate}
\item $P$ is continuous.
\item $A^{\da \ll} = \intr_\sigma(A)$ for any upper subset $A$ of $P$.
\item There exists ${\prec}\in \app(P)$ such that $\lap{A} = \intr_\sigma(A)$ for any upper subset $A$ of $P$.
\item $A^{\ua \ll} = \cl_\sigma(A)$ for any lower subset $A$ of $P$.
\item There exists ${\prec}\in \app(P)$ such that $\uap{A} = \cl_\sigma(A)$ for any lower subset $A$ of $P$.
\end{enumerate}
\end{thm}
\begin{proof}
The equivalence of statements (1), (2) and (3) is a direct result of Lemmata~\ref{lem: cont implies coincidence},~\ref{lem: converse lemma} and~\ref{lem: implication of three conditions}.  As for the equivalence of (2) and (4) (respectively, (3) and (5)), one just invokes Theorem~\ref{thm: uap and lap determines each other} and Lemma~\ref{thm: intr and cl}.
\end{proof}

By Theorem~\ref{thm: cha of cont posets}, we can prove the following result in another way.
\begin{cor}[A characterization of Scott closed sets]
\label{cor: char of scott closed sets}
Let $P$ be a continuous poset, $A \subseteq P$ and $x \in P$.
Then $x \in \cl_\sigma(A)$ if and only if $\dda x \subseteq \da A$.
\end{cor}
\begin{proof}
By Theorem~\ref{thm: cha of cont posets}, $\cl_\sigma (A) = A^{\ua \ll}$.
So $x \in \cl_\sigma(A) \iff s_\ll(x) = \dda x \subseteq \da A$.
\end{proof}

\section{Further properties of the $\mu^\prec$-topology}
\label{sec: further properties}
The $\mu^\prec$-topology induced by an approximating auxiliary relation on a
poset has very tight connections with the Scott topology as revealed in Section~\ref{subsec: approx operators on approx rel}.  We now carry out further investigations into this new topology.

\begin{defi}[c-space]\cite{hoffmann81}
A \emph{c-space} is a topological space $X$ such that, for every $x \in X$, for every open neighbourhood $U$ of $x$, there is a point $y \in U$ such that $x \in \intr(\ua y)$.
\end{defi}

\begin{lem} \label{lem: lemma 1 to c-space theorem}
Let $P$ be a poset and ${\prec}\in \papp(P)$.
If $P_{\mu^\prec}$ is a $c$-space, then $\prec$ is approximating.
\end{lem}
\begin{proof}
Let $x \in P$ be given.  Since $\prec$ is pre-approximating (i.e., $s_\prec(x)$ is directed), it suffices to prove that $\bigsup s_\prec(x) = x$.  To do this, define a subset $D_x$ of $P$ as follows:
\[
D_x := \{y \in P \mid x \in \intr_{\mu^\prec}(\ua y)\}.
\]
Since $P_{\mu^\prec}$ is a c-space, taking the open neighbourhood $U = P$ of $x$, there exists $y \in P$ such that $x \in \intr_{\mu^{\prec}} (\ua y)$.  Hence $D_x \neq \emptyset$.

We now prove that $D_x \subseteq s_\prec(x)$.  For each $y \in D_x$, we have $x \in \intr_{\mu^\prec} (\ua y) = \lap{\left(\intr_{\mu^\prec} (\ua y)\right)}\subseteq \lap{(\ua y)}$.  Hence
$x \in \lap{(\ua y)}$ which then implies that there exists $p \in s_\prec(x) \cap \ua y$, i.e., $y \leq p \prec x$.  Since $\prec$ is an auxiliary relation, it follows that $y \prec x$.  Thus, $D_x \subseteq s_\prec(x)$.

Since $D_x \subseteq s_\prec(x)$, it is immediate that $x$ is an upper bound of $D_x$.
We now show that $\bigsup D_x = x$.  Let $z$ be any upper bound of $D_x$.  We show that $x \leq z$ by showing that every open neighbourhood $U$ of $x$ contains $z$, and using the fact that $\leq$ is the specialization order of $P_{\mu^\prec}$ by virtue of Proposition~\ref{prop: order compatability of mu prec}.  Since $P_{\mu^\prec}$ is a c-space, $x \in U$ implies that there exists $y \in U$ such that $x \in \intr_{\mu^\prec}(\ua y)$, i.e., $y \in D_x$.  Since $z$ is an upper bound of $D_x$, we have $y \leq z$, and since $U$ is an upper set with respect to the specialization order $\leq$, $z \in U$.  Because $s_\prec(x)$ contains a non-empty subset, i.e., $D_x$ whose supremum equals to $x$, and $x$ is an upper bound for $s_\prec(x)$, it follows that $\bigsup s_\prec(x) = x$.  Thus, the proof that $\prec$ is approximating is complete.
\end{proof}

\begin{lem} \label{lem: lemma 2 to c-space theorem}
Let $P$ be a poset and ${\prec}\in \papp(P)$ which satisfies (INT).
Then the following hold:
\begin{enumerate}
\item For each $x \in P$, $s_\succ(x) := \{y \in P \mid x \prec y\}$ is a $\prec$-open set.
\item The collection $\mathcal{B} := \{s_\succ(x) \mid x \in P\}$ is a base for the $\mu^\prec$-topology.
\end{enumerate}
\end{lem}
\begin{proof}
\begin{enumerate}
\item Let $x \in P$ be given.  Clearly, $s_\succ(x)$ is an upper set. Then $s_\succ(x) := \{y \in P \mid x \prec y\}$ is a $\prec$-open set from Theorem~\ref{thm: char of INT}.

\item Let $U$ be a $\prec$-open set and $x \in U$.  Since $U = \lap{U}$, there exists $y \in U$ such that $y \prec x$.  Thus, $x \in s_\succ(y) \subseteq U$.  So, $\mathcal{B}$ forms a base for the $\mu^\prec$-topology.\qedhere
\end{enumerate}
\end{proof}

\begin{thm} \label{thm: c-space theorem}
Let $P$ be a poset and ${\prec}\in \papp(P)$.
\begin{enumerate}
\item If $P_{\mu^\prec}$ is a c-space, then $\prec$ is approximating.
\item If $\prec$ satisfies (INT), then $P_{\mu^\prec}$ is a c-space.
\end{enumerate}
\end{thm}
\begin{proof}
\begin{enumerate}
\item This is just Lemma~\ref{lem: lemma 1 to c-space theorem}.
\item This follows directly from Lemma~\ref{lem: lemma 2 to c-space theorem}(2).\qedhere
\end{enumerate}
\end{proof}

\begin{defi}[Completely distributive lattice] \label{defi: cdl}
Recall that a complete lattice is said to be \emph{completely distributive} if for any set family $(u^i_j)_{j \in J_i}$, one for each $i \in I$, $\biginf_{i \in I} \bigsup_{j \in J_i} u^i_j = \bigsup_{f \in \Pi_{i \in I} J_i} \biginf_{i \in I} u^i_{f(i)}$ holds.
\end{defi}

It was proven in~\cite{hoffmann81} that a $T_0$-space $X$ is a c-space if and only if the lattice of opens, $(\mathcal{O}(X),\subseteq)$ is completely distributive.

\begin{cor} \label{cor: cdl cor}
Let $P$ be a poset and ${\prec}\in \papp(P)$ satisfy (INT).
Then ${\prec}\in \app(P)$ if and only if $\mu^\prec(P)$ is completely distributive.
\end{cor}

\begin{cor} \label{cor: classical c-space thm}
Let $P$ be a poset.  Then $P$ is continuous if and only if the Scott space $P_\sigma$ is a c-space.
\end{cor}
\begin{proof}
$\implies$: If $P$ is continuous, then $\ll$ is approximating.  By Proposition~\ref{prop: cont implies way-below is INT}, $\ll$ satisfies (INT).  Invoking Theorem~\ref{thm: c-space theorem}(1), the Scott space $P_\sigma$ is a c-space.

$\Longleftarrow$: Assume that $P_\sigma$ is a c-space.  We claim that the Scott opens are precisely those upper sets $U$ for which $U = U^{\da \ll}$.  By Proposition~\ref{prop: ll open sets}, it suffices to prove that every Scott open set $U$ (which is of course upper) satisfies the condition $U = U^{\da \ll}$.  Let $U$ be any Scott open set.  Since $P_\sigma$ is a c-space, for each $x \in U$ there exists $y \in U$ such that $x \in \intr_\sigma(\ua y)$.  We claim that $y \ll x$.  Suppose $D$ is a directed subset whose supremum, $\bigsup D$, exists and $\bigsup D \geq x$.  Because $\intr_\sigma(\ua y)$ is upper and $x$ belongs to it, $\bigsup D \in \intr_\sigma(\ua y)$.  Since $\intr_\sigma(\ua y)$ is Scott open, it is inaccessible by directed suprema, i.e., there exists $d \in D$ such that $d \in \intr_\sigma(\ua y) \subseteq \ua y$; whence, $d \geq y$.  This proves that $y \ll x$.  Thus, $x \in U^{\da \ll}$.

Now let $x \in P$ be arbitrary. From the preceding argument, the non-empty set, $D_x$, defined earlier in the proof of Lemma~\ref{lem: lemma 1 to c-space theorem}, can now be rewritten as $D_x = \{y \in P \mid x \in \intr_{\mu^\ll}(\ua y)\}$.  The same argument as in the proof of Lemma~\ref{lem: lemma 1 to c-space theorem} justifies that $D_x \subseteq \dda x$.  Since the interior operator preserves binary intersection, $D_x$ is directed.
By Proposition~\ref{prop: basis lemma}, $\dda x$ must be directed and $\bigsup \dda x = x$, and the proof that $P$ is continuous is thus complete.
\end{proof}

\section{Scott closure in a continuous poset}
\label{sec: scott closure}
In this section, we showcase some domain-theoretic spin-offs that are a direct consequence of findings in the preceding section.  In particular, we recall a characterization of Scott closed sets (c.f. Corollary~\ref{cor: char of scott closed sets}), i.e., the Scott closure of a subset $A$ of a continuous poset comprises precisely those elements way-below which are all members of $\da A$.  We now proceed with a more in-depth study of the structure of $\cl_{\sigma}(A)$, following up this lead.

A subset $A$ of a poset $P$ is Scott closed if and only if it is a lower subset of $P$ and closed under the formation of existing directed suprema.
One might claim that if one takes all the existing directed suprema of subsets of $\da A$ then one should obtain the Scott closure of $A$.

\begin{example} \label{eg: countereg to A' being the Scott closure}
The preceding claim can be easily refuted by the following counterexample $P$: (See Figure~\ref{scott closure}.)
\begin{figure}[tbph!]
\centering{
\includegraphics[width=0.30\textwidth]{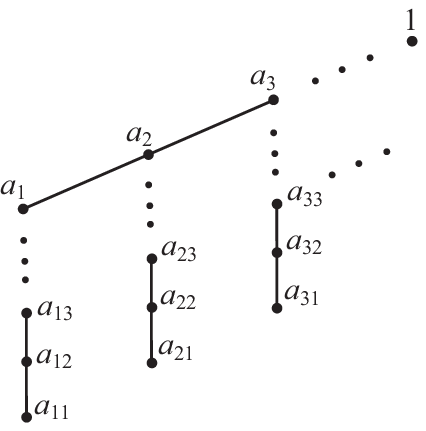}
}
\caption{$P$}
\label{scott closure}
\end{figure}

\noindent Let $A := \{a_{ij},i,j\in \mathbb{N}\}$.  Note that $1 \in \cl_{\sigma}(A) = P$ but there is no directed subset $D$ in ${\downarrow} A$ for which $\bigsup D = 1$.
Thus, taking all the existing directed suprema of $A$ in addition to those elements in $\da A$ does not, in general, form the Scott closure of $A$.
\end{example}

Let $P$ be a poset and $A\subseteq P$, we define
\[
A' := \{x \in P \mid \exists \text{ directed subset } D \text{ of } \da A.~x = \bigsup D\}.
\]

\begin{prop} \label{prop: one step closure of Scott closed set}
Let $P$ be a poset and $A \subseteq P$.  Then, the following hold:
\begin{enumerate}
\item $A \subseteq \da A \subseteq A' \subseteq \cl_{\sigma}(A)$.
\item $A' \subseteq A^{\uparrow\ll}$.
\item $A' = A$ if and only if $A$ is a Scott closed set.
\end{enumerate}
\end{prop}
\begin{proof}
\begin{enumerate}
\item Obvious.
\item Suppose $x \in A'$.  Then there is a directed subset $D$ of $\da A$ such that $\bigsup D = x$. For each $y \ll x$, there is an element $d \in D$ such that $y \leq d$.
    Thus, $y \in \da A$ so that $\dda x \subseteq \da A$.  Hence $x \in A^{\uparrow\ll}$.
\item If $A' = A$, then $A$ is a lower set and closed for existing directed sups which implies that $A$ is a Scott closed set. If $A$ is a Scott closed set, then by (1),
    we have $A \subseteq A' \subseteq \cl_{\sigma}(A) = A$.  Hence $A' = A$.\qedhere
\end{enumerate}
\end{proof}

\begin{defi}[One-step closure] \label{defi: one-step closure}
A poset $P$ is said to have \emph{one-step closure} if $A' = \cl_\sigma(A)$ for every subset $A \subseteq P$.
\end{defi}

\begin{remark} \label{remark: alternative of one-step closure}
Let $P$ be a poset and $A \subseteq P$.  If $A' = \cl_\sigma(A)$, then $A'$ is a Scott closed set.  Conversely, if $A'$ is a Scott closed set, then $\cl_\sigma(A') = A'$.
By Proposition~\ref{prop: one step closure of Scott closed set}(1), $A \subseteq A'$ always holds so that $\cl_\sigma(A) \subseteq \cl_\sigma (A') = A'$.  Again since Proposition~\ref{prop: one step closure of Scott closed set}(1) asserts that $A' \subseteq \cl_\sigma(A)$, we must have $A' = \cl_\sigma(A)$.  Thus, a poset $P$ has one-step closure if and only if $A'$ is Scott closed for every subset $A$ of $P$.
\end{remark}

\begin{thm} \label{thm: continuity implies one-step closure}
Every continuous poset has one-step closure.
\end{thm}
\begin{proof}
Let $P$ be a continuous poset and $A \subseteq P$.  Then, on one hand, $A' \subseteq \cl_{\sigma}(A)$ holds by Proposition~\ref{prop: one step closure of Scott closed set}(1). On the other hand, if $x \in \cl_{\sigma}(A)$, then $\dda x \subseteq \da A$ by virtue of Corollary~\ref{cor: char of scott closed sets}.  But $P$ is continuous so that
$\dda x$ is directed.  These together imply that $x \in A'$.  Thus, we can conclude that $\cl_{\sigma}(A) = A'$.
\end{proof}

\begin{remark}
At the point of writing, it was brought to the attention of the authors that Theorem~\ref{thm: continuity implies one-step closure} had earlier been established by Dongsheng Zhao in his Ph.D. thesis (see \cite[Lemma 2.5]{zhao92}) but the argument we use here has a very different motivation from Zhao's.
\end{remark}

What order-theoretic properties does a poset having one-step closure enjoy?
To answer this question, we need to recall the notion of meet-continuity of posets which was  first introduced by~\cite{kouetal01}.
\begin{defi} \label{defi: meet-continuity}
A poset $P$ is called \emph{meet continuous} if for any $x \in P$ and any directed subset $D$, whenever $\bigsup D$ exists and $x \leq \bigsup D$, then
$x \in \cl_{\sigma}(\da D \cap \da x)$.
\end{defi}

\begin{thm} \label{thm: one-step closure implies meet-continuity}
All posets having one-step closure are meet-continuous.
\end{thm}
\begin{proof}
Suppose that there is a directed set $D \subseteq P$ whose supremum exists and $x \in P$ is such that $x \leq \bigsup D$.  Then $\bigsup D \in D'$. Since $P$ has one-step closure, it follows that $D'$ is a Scott closed set which then implies that $x \in D'$.  So, there exists a directed set $D_1 \subseteq \da D$ such that $\bigsup D_1 = x$.  It is clear that $D_1 \subseteq \da D \cap \da x$ and hence $x \in (\da D \cap \da x)' = \cl_\sigma(\da D \cap \da x)$.  This then completes the proof that $P$ is meet continuous.
\end{proof}

\begin{thm} \label{thm: one-step closure implies Scott int is way-above}
Let $P$ be a poset that has one-step closure.
Then $\intr_\sigma(\ua x) = \dua x$ for each $x \in P$.
\end{thm}
\begin{proof}
Since $\intr_\sigma (\ua x) \subseteq \dua x$ always holds for any element $x \in P$, it remains to prove the reverse containment.  For the sake of contradiction, suppose that
there is $y \in \dua x \setminus \intr_\sigma (\ua x)$.  Then
$y \in P \setminus \intr_\sigma(\ua x) = \cl_{\sigma}(P \setminus \ua x)
= (P \setminus \ua x)'$.  So there is a directed subset
$D \subseteq \da (P \setminus \ua x) = P \setminus \ua x$ such that
$\bigsup D = y$.  This will run contrary to $x \ll y$.
Hence $\dua x \subseteq \intr_\sigma( \ua x)$.
\end{proof}

\begin{remark}
Theorems~\ref{thm: one-step closure implies meet-continuity} and~\ref{thm: one-step closure implies Scott int is way-above} give us a more fine-grained means to explain the phenomena that continuous posets are always meet-continuous and possess one-step closure.
\end{remark}

At the point of writing, we are not able to completely characterize those poset that have a one-step closure.

\section{Conclusion}
\label{sec: conclusion}
In this paper, we develop the core of domain theory using suitably defined approximation operators.

In domain theory, the role played by auxiliary relations is indeed quite auxiliary -- once the continuity of a dcpo $L$ is characterized by the condition that the way-below relation $\ll$ is the smallest approximating auxiliary relation on $L$, the auxiliary relations are almost never mentioned in the subsequent development of domain theory.  For instance, auxiliary relations are involved in the study of the Scott topology (which is of central importance in domain theory).  Our approximation operators fill in this gap by exploiting the lower and upper approximation operators induced by auxiliary relations.  Many salient domain-theoretic concepts such as interpolation property (INT), Scott topology, etc., take on new meanings in the light of operator theory.  Using our approximation operators, not only can several well-known results regarding the way-below relation for continuous posets be generalized to the context of approximating relations, new domain-theoretic properties regarding continuous posets have also been obtained.

While it is theoretically pleasing to see that the core of domain theory can be formulated and developed using operator theory, one can perhaps look for more practical applications derived from the matrimony of domain theory and operator theory.  In real life situations, there may be opportunities to export the domain-theoretic considerations derived in this paper back into the world of imprecise knowledge.

We end this paper with some open problems:
\begin{question}\leavevmode
\begin{enumerate}
\item Characterize those posets having a one-step closure.
\item Quasicontinuous posets are a generalization of continuous posets and are gaining much research attention recently~\cite{gierzetal83}.  Can our operator theoretic approach in this paper be employed to study quasi-continuous posets?
\end{enumerate}
\end{question}

\end{document}